%% file: main.tex
\documentclass[journal]{IEEEtran}

\usepackage{amsmath,amssymb,amsthm,mathtools}
\usepackage{algorithm}
\usepackage{algorithmic}
\usepackage{booktabs}
\usepackage{cite}
\usepackage{enumitem}
\usepackage{graphicx}
\usepackage{tikz}
\usepackage{xcolor}
\usepackage{url}
\usepackage[hidelinks]{hyperref}

\usetikzlibrary{arrows.meta,positioning,calc,fit,backgrounds}

\definecolor{tnsmblue}{HTML}{1F4E79}
\definecolor{tnsmred}{HTML}{B22222}
\definecolor{tnsmgray}{HTML}{555555}

\newtheorem{theorem}{Theorem}[section]
\newtheorem{proposition}[theorem]{Proposition}
\newtheorem{lemma}[theorem]{Lemma}
\newtheorem{definition}[theorem]{Definition}
\newtheorem{remark}[theorem]{Remark}

\begin{document}

\title{Security-Induced Braess Paradoxes in Service Function Chain Orchestration}

\author{Daniel~Commey
and~Bin~Mai%
\thanks{D.~Commey is with the Department of Computer Engineering and Computer Science, California State University, Long Beach, Long Beach, CA 90840, USA (e-mail: daniel.commey@csulb.edu).}%
\thanks{B.~Mai is with the College of Business, Western Carolina University (e-mail: bmai@wcu.edu).}}

\maketitle

\begin{abstract}
NFV/SDN orchestration lets operators instantiate and steer traffic through virtual firewalls, IDS/IPS replicas, WAF clusters, zero-trust gateways, backup inspection paths, and migration targets on demand. Operators often treat these options as monotone improvements: more inspection capacity, lower nominal latency, or broader placement flexibility should not degrade the service. That intuition can fail even when the new option is locally attractive. We study a security-induced Braess paradox in service function chain (SFC) orchestration, where adding a defensive option worsens the post-adaptation equilibrium by concentrating traffic and adversarial value on shared security resources. We define Braessian security-management actions, derive a sufficient condition for paradox emergence under affine load-dependent VNF delay, and give a pre-deployment orchestration screen that rejects, caps, or reserves harmful options. A multi-tenant SFC experiment suite applies the model to four topology-derived settings: a fat-tree datacenter, NSFNET-style WAN, GEANT-style WAN, and edge/fog topology. Under default parameters in the Braessian regime identified by the theory, naive defensive expansion raises equilibrium service cost by 27.2--30.8\% and increases risk concentration by factors of 6.1--9.7. Paradox-aware constrained use keeps the residual penalty below 1.9\%, reduces service cost by 20.0--22.1\% relative to naive expansion, and lowers a concentration-sensitive attack-loss proxy by 93.5\% on average.
\end{abstract}


\section{Introduction}

\IEEEPARstart{N}{etwork} security has become an orchestration problem. Rather than deploy fixed middleboxes at fixed choke points, operators instantiate and steer traffic through virtual firewall (FW) instances, intrusion detection system (IDS) replicas, web application firewall (WAF) clusters, zero-trust gateways, moving-target migrations, and backup inspection paths. The Service Function Chaining (SFC) architecture formalizes this steering model by treating a network service as an ordered traversal of service functions \cite{rfc7665}. Network Function Virtualization and Software-Defined Networking (NFV/SDN) orchestration then decides where functions should run and which service chains should use them \cite{etsi2014mano,sun2018sfc,bari2016orchestrating}.

Standard engineering practice assumes monotonicity: if a security option improves a local metric, adding it should not make the managed service worse. A faster IDS replica, a central zero-trust gateway, or a backup inspection tunnel looks beneficial because it lowers free-flow cost, increases coverage, or expands feasible placements. We study when this intuition fails.

The mechanism differs from ordinary security overhead. Forcing every packet through an expensive inspection function raises latency directly. Here the added option is optional and locally attractive. The degradation appears only after traffic, tenants, orchestrators, or attackers adapt to its availability: the feasible action space changes, and the new equilibrium concentrates load and adversarial value on a shared security resource.

This non-monotonicity is analogous to the classical Braess paradox in congested routing, where adding a road can increase equilibrium travel time \cite{braess2005paradox,wardrop1952theoretical,beckmann1956studies,roughgarden2002selfish,roughgarden2005selfish}. In our setting, the added edge is a defensive resource---a virtual network function (VNF), gateway, inspection shortcut, backup path, or migration target---and the harm appears as degraded service-chain latency, higher overload, greater blast radius, or larger expected adversarial loss.

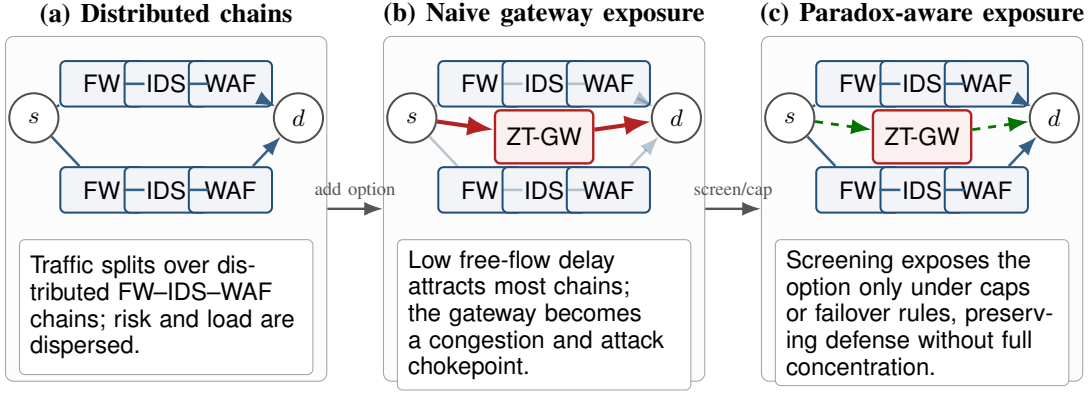
\begin{figure*}[t]
\centering
\begin{tikzpicture}[
  font=\sffamily\small,
  >=Latex,
  panel/.style={draw=black!35, rounded corners=3pt, fill=black!1, inner sep=6pt},
  endpoint/.style={circle, draw=black!70, fill=white, minimum size=6.5mm, line width=0.75pt},
  vnf/.style={draw=tnsmblue, fill=tnsmblue!6, rounded corners=2pt, minimum width=11mm, minimum height=6mm, line width=0.75pt},
  gateway/.style={draw=tnsmred, fill=tnsmred!7, rounded corners=2pt, minimum width=13mm, minimum height=7mm, line width=0.9pt},
  note/.style={draw=black!35, fill=white, rounded corners=2pt, align=left, inner sep=4pt, text width=3.6cm, anchor=north, minimum height=1.78cm},
  flow/.style={->, line width=0.95pt, tnsmblue!90},
  heavy/.style={->, line width=1.55pt, tnsmred},
  cappedflow/.style={->, line width=1.05pt, green!45!black, dashed}
]

\node[panel, minimum width=4.25cm, minimum height=4.55cm, label={[font=\bfseries]above:{(a) Distributed chains}}] (p1) at (0,-0.1) {};
\node[panel, minimum width=4.25cm, minimum height=4.55cm, label={[font=\bfseries]above:{(b) Naive gateway exposure}}] (p2) at (5.0,-0.1) {};
\node[panel, minimum width=4.25cm, minimum height=4.55cm, label={[font=\bfseries]above:{(c) Paradox-aware exposure}}] (p3) at (10.0,-0.1) {};

\foreach \x/\prefix in {0/a,5/b,10/c}{
  \node[endpoint] (\prefix s) at (\x-1.75,1.1) {$s$};
  \node[endpoint] (\prefix d) at (\x+1.75,1.1) {$d$};
  \node[vnf] (\prefix fw1) at (\x-0.85,1.55) {FW};
  \node[vnf] (\prefix ids1) at (\x,1.55) {IDS};
  \node[vnf] (\prefix waf1) at (\x+0.85,1.55) {WAF};
  \node[vnf] (\prefix fw2) at (\x-0.85,0.15) {FW};
  \node[vnf] (\prefix ids2) at (\x,0.15) {IDS};
  \node[vnf] (\prefix waf2) at (\x+0.85,0.15) {WAF};
}

\draw[flow] (as) -- (afw1) -- (aids1) -- (awaf1) -- (ad);
\draw[flow] (as) -- (afw2) -- (aids2) -- (awaf2) -- (ad);
\node[note] at (0,-0.5) {Traffic splits over distributed FW--IDS--WAF chains; risk and load are dispersed.};

\node[gateway] (bgw) at (5.0,0.85) {ZT-GW};
\draw[flow, opacity=0.35] (bs) -- (bfw1) -- (bids1) -- (bwaf1) -- (bd);
\draw[flow, opacity=0.35] (bs) -- (bfw2) -- (bids2) -- (bwaf2) -- (bd);
\draw[heavy] (bs) -- (bgw);
\draw[heavy] (bgw) -- (bd);
\node[note] at (5,-0.5) {Low free-flow delay attracts most chains; the gateway becomes a congestion and attack chokepoint.};

\node[gateway] (cgw) at (10.0,0.85) {ZT-GW};
\draw[flow] (cs) -- (cfw1) -- (cids1) -- (cwaf1) -- (cd);
\draw[flow] (cs) -- (cfw2) -- (cids2) -- (cwaf2) -- (cd);
\draw[cappedflow] (cs) -- (cgw);
\draw[cappedflow] (cgw) -- (cd);
\node[note] at (10,-0.5) {Screening exposes the option only under caps or failover rules, preserving defense without full concentration.};

\draw[->, line width=0.9pt, tnsmgray] (p1.east) -- node[above, font=\scriptsize, text=tnsmgray] {add option} (p2.west);
\draw[->, line width=0.9pt, tnsmgray] (p2.east) -- node[above, font=\scriptsize, text=tnsmgray] {screen/cap} (p3.west);
\end{tikzpicture}
\caption{Security-induced Braess effect in SFC orchestration. The added zero-trust gateway is locally attractive, but unrestricted exposure shifts the post-adaptation equilibrium toward a shared security resource. Paradox-aware orchestration treats the gateway as an equilibrium-shaping resource and limits ordinary use when unrestricted exposure is harmful.}
\label{fig:sfc-overview}
\end{figure*}

Fig.~\ref{fig:sfc-overview} illustrates the pattern. Before intervention, tenant chains are distributed across multiple FW--IDS--WAF sequences. The operator adds a zero-trust gateway with low generalized free-flow cost. Naive orchestration exposes this option to all tenants. Traffic then concentrates on the new gateway, which couples the load-sensitive ingress and egress security planes and becomes an attacker-preferred chokepoint.

The main contributions are:
\begin{itemize}[leftmargin=*,topsep=2pt,itemsep=2pt]
    \item \textbf{Failure-mode formulation.} Security-induced Braess paradoxes are defined as management actions that expand defensive feasibility but worsen post-adaptation service cost or adversarial loss.
    \item \textbf{SFC equilibrium model.} Multi-tenant SFC routing is modeled over shared link and VNF resources with affine load-dependent delay, risk exposure, and Wardrop equilibrium.
    \item \textbf{Structural condition and mitigation.} We derive a sufficient condition for a defensive shortcut to be Braessian and give a paradox-aware screening algorithm that rejects or caps harmful options.
    \item \textbf{Operational baselines.} We compare unrestricted exposure against load-aware caps, risk-aware scoring, min--max utilization control, and a marginal-cost system optimum.
    \item \textbf{Topology-derived evaluation.} We implement the model on datacenter, WAN, and edge/fog topologies with multiple tenant requests and summary figures. Naive expansion increases equilibrium cost by up to 30.8\%, while paradox-aware constrained use reduces the naive penalty by about an order of magnitude.
\end{itemize}

\section{Related Work}

\subsection{Braess Effects and Selfish Routing}

Wardrop's equilibrium principle states that used routes have no higher perceived cost than unused alternatives \cite{wardrop1952theoretical}. Beckmann et al.\ gave the convex potential formulation for separable congestion games \cite{beckmann1956studies}. Braess introduced the paradox in transportation planning \cite{braess1968paradoxon}, later translated and revisited in \cite{braess2005paradox}; subsequent work developed the paradox and the price of anarchy in general networks \cite{murchland1970braess,roughgarden2002selfish,roughgarden2005selfish}. Braess-type effects have also been studied directly in communication and computing systems, including non-cooperative routing in computer networks \cite{korilis1999avoiding}, loss networks \cite{bean1997loss}, and wireless technology upgrades whose equilibria can worsen despite higher nominal capacity \cite{dinitz2013wireless}. The setting here differs in the object being added: the new resource is a defensive function or orchestration choice, not a transportation link or ordinary communication link.

\subsection{SFC/NFV Orchestration}

SFC decomposes network services into ordered service functions and steering rules \cite{rfc7665}. NFV and SDN make these chains programmable by decoupling network functions from proprietary appliances and exposing control-plane steering mechanisms \cite{etsi2014mano,kreutz2015sdn,mijumbi2016nfv}. SFC orchestration research addresses placement, routing, multi-domain coordination, and resource constraints \cite{mehraghdam2014chains,cohen2015placement,bari2016orchestrating,sun2018sfc,beck2017coordvnf}. Surveys further organize the design space around traffic steering, lifecycle management, and embedding constraints \cite{medhat2017sfc,hantouti2019traffic,bhamare2016survey}. In that setting, the choice of placement is only part of the problem; adding a plausible security option can also change the induced equilibrium in a harmful way.

\subsection{Security Orchestration and Moving Target Defense}

Moving target defense and policy-driven security orchestration use programmability to adapt the attack surface. Yoon et al.\ propose attack-graph-based moving target defense in SDN and evaluate shuffling decisions by attack success and cost \cite{yoon2020attackgraph}; Cho et al.\ survey proactive and adaptive MTD techniques \cite{cho2020mtdsurvey}. Semantic-aware SDN/NFV security orchestration similarly aims to enforce virtual security functions while respecting policy consistency and QoS \cite{zarca2020semantic}. These lines of work motivate automated defense; they do not treat the introduction of a defensive option itself as a possible source of equilibrium degradation.

\subsection{Gap Relative to Prior Work}

Prior SFC/NFV work generally asks how to select a feasible placement or route once the action space has been specified. Prior security-orchestration work generally asks how to improve attack resistance, policy compliance, or response cost under a selected defense model. Between those layers sits a monotonicity assumption: expanding the defensive action space should not hurt if the orchestrator can choose among the added options. In adaptive service systems, a defensive option changes the local path score and the equilibrium induced by the service-management layer. The relevant question is therefore not whether an IDS, gateway, or migration target is useful in isolation, but whether exposing it to normal orchestration worsens the induced service state.

Three questions follow. First, can a locally attractive security-management action be service-Braessian in an SFC system? Second, can an operator detect the effect with pre-deployment topology, demand, VNF capacity, and risk estimates? Third, can the operator preserve some defensive availability without accepting the unrestricted equilibrium penalty?

\section{System Model}
\label{sec:model}

We model a managed service system as
\[
\mathcal{M}=(G,\mathcal{R},\mathcal{F},\mathcal{P},\ell,\rho),
\]
where $G=(V,E)$ is the substrate topology, $\mathcal{R}$ is a set of tenant service requests, $\mathcal{F}$ is the set of deployed service and security functions, $\mathcal{P}$ is the set of feasible SFC paths, $\ell$ gives load-dependent resource delay, and $\rho$ gives resource risk exposure.

\begin{table*}[t]
\centering
\caption{Notation used in the SFC Braess model}
\label{tab:notation}
\begin{tabular}{p{0.14\textwidth}p{0.78\textwidth}}
\toprule
Symbol & Meaning \\
\midrule
$G=(V,E)$ & Substrate topology with node set $V$ and link set $E$. \\
$\mathcal{R}$, $r$, $d_r$ & Tenant request set, request index, and demand of request $r$. \\
$\mathcal{F}$ & Deployed service and security functions, including FW, IDS, WAF, and gateways. \\
$\mathcal{P}_r$, $p$ & Feasible SFC paths for request $r$ and individual path index. \\
$f_p$, $f^\star$ & Flow assigned to path $p$ and post-adaptation equilibrium flow. \\
$e$, $y_e(f)$ & Link or VNF resource and aggregate load on resource $e$. \\
$b_e$, $a_e$, $u_e$ & Fixed delay, load-sensitivity coefficient, and capacity of resource $e$. \\
$\ell_e(y)$, $c_p(f)$ & Resource delay and end-to-end path cost under flow $f$. \\
$\rho_e$, $\chi_e$ & Risk weight and exposure multiplier used in post-equilibrium attack evaluation. \\
$\Delta$, $\mathcal{M}+\Delta$ & Security-management action and post-intervention managed system. \\
$J(\mathcal{M})$, $\Pi(\Delta)$ & Equilibrium average service cost and normalized paradox penalty. \\
$\mathrm{SBR}$, $G_A$ & Security Braess Ratio and Aware gain relative to naive expansion. \\
$A(f)$, $S_{\mathrm{DDoS}}(f)$ & Expected targeted attack-loss proxy and deterministic DDoS stress-loss proxy. \\
$\tau$, $\kappa$, $\mathcal{K}$ & Screening threshold, per-request cap fraction, and cap-grid searched by the orchestrator. \\
\bottomrule
\end{tabular}
\end{table*}

Each request $r\in\mathcal{R}$ has demand $d_r$ and feasible chain set $\mathcal{P}_r$. A path $p\in\mathcal{P}_r$ is an ordered sequence of substrate links and VNF resources. Let $f_p\ge 0$ be the equilibrium flow on path $p$. Feasible flows satisfy
\[
\sum_{p\in\mathcal{P}_r} f_p = d_r,\qquad \forall r\in\mathcal{R}.
\]
For resource $e$, load is
\[
y_e(f)=\sum_{p:e\in p} f_p.
\]
Resources include substrate links, firewall pools, IDS replicas, WAF clusters, and shared gateways. Each resource has affine delay
\[
\ell_e(y_e)=b_e+a_e\frac{y_e}{u_e},
\]
where $b_e$ is fixed delay, $a_e$ is load sensitivity, and $u_e$ is capacity. The service cost of path $p$ is
\[
c_p(f)=\sum_{e\in p}\ell_e(y_e(f)).
\]

\begin{definition}[Multi-tenant SFC equilibrium]
A feasible flow $f^\star$ is an SFC Wardrop equilibrium if for every request $r$ and paths $p,q\in\mathcal{P}_r$,
\[
f_p^\star>0 \Rightarrow c_p(f^\star)\le c_q(f^\star).
\]
\end{definition}

\begin{lemma}[Beckmann equivalence for SFC equilibrium]
\label{lem:beckmann}
With nondecreasing separable resource delays, a feasible flow is a multi-tenant SFC Wardrop equilibrium if and only if it solves the Beckmann program
\begin{align}
\min_f \quad & \Phi(f)=\sum_e \int_0^{y_e(f)}\ell_e(z)\,dz \\
\text{s.t.}\quad & \sum_{p\in\mathcal{P}_r} f_p=d_r,\quad f_p\ge 0.
\end{align}
\end{lemma}

\begin{IEEEproof}
The feasible set is a product of request-simplexes and is convex and compact. Because $y_e(f)$ is linear in $f$ and each $\ell_e(\cdot)$ is nondecreasing, $\Phi$ is convex. The derivative of $\Phi$ with respect to path flow $f_p$ is
\[
\frac{\partial \Phi}{\partial f_p}
=\sum_{e\in p}\ell_e(y_e(f))=c_p(f).
\]
The Karush--Kuhn--Tucker conditions for each request simplex therefore state that any path with positive flow has minimum path cost among that request's feasible paths, while any unused path has cost no smaller than the request's multiplier. These are exactly the Wardrop conditions. Conversely, any Wardrop flow satisfies these KKT conditions and therefore minimizes the convex potential.
\end{IEEEproof}

Lemma~\ref{lem:beckmann} is the standard Beckmann equivalence \cite{beckmann1956studies} specialized to SFC path--resource structure. We state it explicitly because (i) strict convexity of $\Phi$ under affine delays guarantees a unique equilibrium load vector, and (ii) the convex program gives a closed-form way to compute the post-adaptation flow.

\subsection{Wardrop Versus System Optimum}

Wardrop equilibrium describes settings in which tenants, traffic classes, or local orchestration policies react to path costs without fully internalizing the congestion externality imposed on other chains. Such behavior can arise in multi-tenant NFV/SDN systems when chain selection is split across domains, service controllers, or policy layers. To separate equilibrium effects from the physical presence of the added gateway, we also report a \emph{system-optimum expansion} baseline in which a centralized controller minimizes total delay after the gateway is added. For affine delay $\ell_e(y)=b_e+a_e y/u_e$, the system optimum is obtained by solving a Wardrop problem with marginal delay $b_e+2a_e y/u_e$ and evaluating the resulting flow under the original delays. If the system optimum avoids the penalty while Wardrop does not, the added gateway is harmful because of adaptive equilibrium use, not because its presence makes the service infeasible.

\begin{proposition}[Marginal-cost system optimum]
\label{prop:system-optimum}
For affine resource delays $\ell_e(y)=b_e+a_e y/u_e$, any Wardrop equilibrium computed with modified delays $\tilde{\ell}_e(y)=b_e+2a_e y/u_e$ minimizes total service delay $\sum_e y_e\ell_e(y_e)$ over the same feasible SFC flow set.
\end{proposition}

\begin{IEEEproof}
The total-delay objective can be written as
\[
T(f)=\sum_e y_e(f)\ell_e(y_e(f))
=\sum_e \left(b_e y_e(f)+\frac{a_e}{u_e}y_e(f)^2\right).
\]
Because $a_e\ge 0$ and $y_e(f)$ is linear in $f$, $T(f)$ is convex. Its path derivative is
\[
\frac{\partial T}{\partial f_p}
=\sum_{e\in p}\left(b_e+2a_e y_e(f)/u_e\right)
=\sum_{e\in p}\tilde{\ell}_e(y_e(f)).
\]
The KKT conditions for minimizing $T(f)$ over the request-simplex constraints are therefore identical to the Wardrop conditions under modified delay functions $\tilde{\ell}_e$. Hence any Wardrop solution for the modified delays is a system optimum for the original affine-delay model.
\end{IEEEproof}

\subsection{Security Intervention}

A security-management action $\Delta$ modifies the feasible system by adding or exposing a defensive option: a VNF, zero-trust gateway, backup inspection route, migration target, or security shortcut. We write the modified system as $\mathcal{M}'=\mathcal{M}+\Delta$.

For generality, define a free-flow generalized cost $g_p^0=\sum_{e\in p}b_e-\lambda_s\sigma_p$, where $\sigma_p\ge 0$ is a local path security benefit score and $\lambda_s\ge 0$ converts inspection benefit into delay units. A new defensive path is \emph{locally attractive} for request $r$ if its free-flow generalized cost is no larger than that of at least one existing path in $\mathcal{P}_r$. All experiments set $\lambda_s=0$ to isolate delay-driven attractiveness; explicit security-benefit scoring can be added by increasing $\lambda_s$ without changing the equilibrium-screening framework.

The operator's service objective is the demand-weighted average equilibrium cost
\[
J(\mathcal{M})=\frac{1}{\sum_r d_r}\sum_{p\in\mathcal{P}} f_p^\star c_p(f^\star).
\]
For adversarial impact, each resource has risk weight $\rho_e\ge 0$ and exposure multiplier $\chi_e\ge 0$. We use a concentration-sensitive attack-loss proxy
\[
A(f)=\max_e \rho_e\chi_e\left(\frac{y_e(f)}{\sum_r d_r}\right)^2\left(1+\left[\frac{y_e(f)}{u_e}-1\right]_+\right)\sum_r d_r,
\]
where $[x]_+=\max(0,x)$. The proxy reflects the assumption that an adaptive attacker targets high-flow, high-exposure, overloaded chokepoints.

The adversary is not jointly optimized with routing in this paper. Service equilibrium is computed first; only after the equilibrium flow is fixed do we evaluate the resource an adaptive attacker would prefer under $A(f)$. This sequential model matches a setting in which the service controller publishes or exposes an SFC policy and an adversary subsequently probes the induced traffic pattern. It also separates the traffic/orchestration non-monotonicity from a defender-attacker game; the latter is a useful extension but is not required for the Braessian service-cost effect.

\begin{definition}[Security-induced Braess action]
A security-management action $\Delta$ is service-Braessian if
\[
J(\mathcal{M}+\Delta)>J(\mathcal{M})
\]
even though $\Delta$ introduces a locally attractive defensive path under $g_p^0$. It is adversarially Braessian if
\[
A(f^\star_{\mathcal{M}+\Delta})>A(f^\star_{\mathcal{M}}).
\]
\end{definition}

We report the normalized paradox penalty
\[
\Pi(\Delta)=\frac{J(\mathcal{M}+\Delta)-J(\mathcal{M})}{J(\mathcal{M})}
\]
and the Security Braess Ratio
\[
\mathrm{SBR}(\Delta)=\frac{J(\mathcal{M}+\Delta)}{J(\mathcal{M})}.
\]
For a mitigation policy $h$ compared with naive unrestricted expansion, we define
\[
G_A(h)=\frac{J_{\mathrm{naive}}-J_h}{J_{\mathrm{naive}}},
\]
and call this quantity the \emph{Aware gain} when $h$ is the paradox-aware policy.

The Distributed Denial of Service (DDoS) stress-loss metric used in the evaluation is computed after equilibrium by selecting
\[
e^\star\in\arg\max_e \left(\frac{y_e(f^\star)}{u_e},\, y_e(f^\star),\, \chi_e\right)
\]
lexicographically and reporting
\begin{align}
S_{\mathrm{DDoS}}(f^\star)
&=\eta_D\,\theta_{e^\star}(f^\star)\,
\left(1+\left[\theta_{e^\star}(f^\star)U/u_{e^\star}-1\right]_+\right) \nonumber\\
&\quad \times \left(1+\rho_{e^\star}\chi_{e^\star}\right),
\end{align}
where $U=\sum_r d_r$ and $\theta_e(f)=y_e(f)/U$,
with $\eta_D=1.25$ in all experiments. The multiplier $\eta_D$ only fixes the scale of the proxy and does not affect policy comparisons; at full demand share and full utilization the proxy evaluates to $\eta_D(1+\rho_{e^\star}\chi_{e^\star})$. Lexicographic selection by utilization, load, and exposure yields a conservative post-equilibrium adversary: it targets the most stressed resource first, breaking ties by affected traffic and then by exposed value.

\subsection{Operational Interpretation}

The model represents the policy layer of an NFV/SDN security orchestrator. A candidate action $\Delta$ may correspond to exposing a new gateway in a service-function classifier, adding a next-hop group in an SDN controller, advertising a backup inspection chain to tenants, or allowing a new VNF placement site in a MANO workflow. The action can be physically deployed while still being hidden, capped, or reserved from normal chain selection. The paradox is driven by how the defensive asset changes the feasible set used by adaptive traffic and policy controllers, not by the asset's mere existence.

The Wardrop abstraction is a post-policy steady state. It fits federated edge/fog deployments, multi-domain SFC, tenant-selectable service classes, and repeated local re-optimization where no single policy layer prices the congestion externality imposed on shared security resources. The fat-tree experiment serves as a datacenter stress case: it asks what happens when a nominally centralized environment still exposes a shortcut to cost-driven chain selection without marginal congestion pricing. The system-optimum baseline in Section~\ref{sec:evaluation} represents the opposite extreme, a centralized controller that internalizes the externality.

\subsection{Running Example}

Before presenting the general condition, consider a numerical SFC motif with two distributed chains. Let $c=2$, $\alpha=1$, and $D=1$. Without the gateway, the equilibrium splits traffic equally between the two distributed chains. Each chain pays fixed distributed inspection delay $2$ plus load-sensitive delay $1/2$, so the equilibrium service cost is $C_0=2.5$.

Now suppose the operator exposes a zero-trust gateway path with fixed delay $\epsilon=0.75$ that traverses both shared security planes. In free flow, the gateway path is attractive because $0.75<2$. A local evaluation would therefore mark it as a faster security option. After all traffic shifts to the gateway, however, both shared planes carry the entire demand and the gateway path cost is
\[
C_s=\epsilon+2\alpha=2.75.
\]
No infinitesimal user can improve by returning to a distributed chain, because that chain still traverses one of the fully loaded shared planes and costs $c+\alpha=3$. The new equilibrium is therefore stable and strictly worse than the original. The failure is not that the operator added mandatory overhead; it is that an attractive option, once used at equilibrium, coupled resources that were previously separated.

The example also explains why simply adding more security capacity is not the right abstraction. If the gateway were centrally scheduled with marginal congestion prices, the controller could keep some traffic on distributed chains and avoid the penalty. The paradox appears when the defensive option is advertised into a selection process that does not internalize the congestion and risk concentration it creates.

\section{Structural Result}
\label{sec:theory}

The following result gives a sufficient condition for an added defensive shortcut to be Braessian in an SFC motif. The condition is stylized to isolate the mechanism that the topology experiments instantiate at scale.

\begin{theorem}[Braessian defensive shortcut]
\label{thm:braess}
Consider one aggregate service class with demand $D$. Before intervention, it has two feasible security chains. Chain 1 uses a load-sensitive ingress security resource with delay $\alpha y/D$ and then a fixed distributed inspection resource of delay $c$. Chain 2 uses a fixed distributed inspection resource of delay $c$ and then a load-sensitive egress security resource with delay $\alpha y/D$. The pre-intervention equilibrium cost is $C_0=c+\alpha/2$.

Now add a defensive shortcut of fixed delay $\epsilon$ that traverses both load-sensitive security resources. If
\[
c-\frac{3\alpha}{2}<\epsilon\le c-\alpha,
\]
then full use of the new shortcut is the unique Wardrop equilibrium after intervention and the equilibrium cost is strictly higher than $C_0$.
\end{theorem}

\begin{IEEEproof}
Before the shortcut is added, let $x$ be the demand on Chain~1 and $D-x$ be the demand on Chain~2. The two path costs are
\[
C_1(x)=\alpha x/D+c,\qquad
C_2(x)=c+\alpha(D-x)/D.
\]
At equilibrium, any two used paths must have equal cost. Since both chains are feasible and costs are continuous and increasing in their own load-sensitive component, the unique split satisfying $C_1(x)=C_2(x)$ is $x=D/2$. Thus the pre-intervention equilibrium cost is $C_0=c+\alpha/2$.

After adding the shortcut, suppose all demand uses it. Then both load-sensitive resources carry load $D$, so the shortcut cost is $C_s=\epsilon+2\alpha$. If an infinitesimal user deviates to either original chain while all other users remain on the shortcut, that deviating chain uses exactly one load-sensitive resource already loaded by $D$ and one fixed distributed inspection resource, so the deviating cost is $c+\alpha$. The condition $\epsilon\le c-\alpha$ implies $C_s=\epsilon+2\alpha\le c+\alpha$, so no user can strictly improve by deviating. Hence full shortcut use is a Wardrop equilibrium. Since $\epsilon \le c - \alpha < c$, the shortcut's free-flow cost is strictly below that of both existing chains, confirming local attractiveness. The Beckmann potential is strictly convex in the loads of the two load-sensitive resources (the fixed-delay resources enter linearly), so those equilibrium loads are unique; since the shortcut is the only path using both load-sensitive resources, the load pair determines the shortcut flow and hence all path flows, and the full-shortcut assignment identified above is the only post-intervention equilibrium. The condition $\epsilon>c-3\alpha/2$ implies $C_s=\epsilon+2\alpha>c+\alpha/2=C_0$, so this unique equilibrium is strictly worse than the original.
\end{IEEEproof}

\begin{remark}
The interval in Theorem~\ref{thm:braess} is nonempty whenever $\alpha>0$. If the shortcut delay must be nonnegative, feasible $\epsilon$ values exist when $c-\alpha\ge 0$ and $\max\{0,c-3\alpha/2\}<c-\alpha$. Thus a nonnegative locally attractive shortcut exists, for example, whenever $c>\alpha$. The theorem does not say that all security gateways are harmful. It states a management condition: when a locally attractive defensive option couples multiple load-sensitive security resources, the new equilibrium can be worse than the original feasible set.
\end{remark}

\begin{proposition}[SFC shared-resource extension]
\label{prop:sfc-extension}
Consider $m\ge 2$ symmetric distributed SFC paths for one aggregate demand $D$. Each path has one load-sensitive security resource with delay $\alpha y/D$ and fixed distributed inspection delay $c$. Before intervention, the unique symmetric equilibrium has cost $C_0=c+\alpha/m$. Now add a defensive chain with fixed delay $\epsilon$ that traverses all $m$ load-sensitive security resources. If
\[
c+\frac{\alpha}{m}<\epsilon+m\alpha\le c+\alpha,
\]
then full use of the defensive chain is a Wardrop equilibrium and its equilibrium cost is strictly larger than the pre-intervention cost.
\end{proposition}

\begin{IEEEproof}
Before intervention, symmetry and strict monotonicity of each distributed load-sensitive resource imply that the unique equilibrium load assigns $D/m$ to each distributed chain, giving cost $C_0=c+\alpha/m$. After the defensive chain is added, suppose all demand uses it. Each distributed load-sensitive resource carries load $D$, so the defensive-chain cost is $C_s=\epsilon+m\alpha$. If an infinitesimal user deviates to any distributed chain, that chain's load-sensitive resource is already loaded by $D$, and the deviating cost is $c+\alpha$. The upper-bound condition $C_s\le c+\alpha$ therefore rules out any strictly improving deviation. It also implies $\epsilon\le c-(m-1)\alpha<c$, so the defensive chain is locally attractive in free-flow delay. The lower-bound condition $C_s>C_0$ makes the post-intervention equilibrium strictly more costly than the pre-intervention equilibrium. This proves a Braessian SFC condition whenever a locally attractive defensive chain couples the security resources that were previously separated across distributed chains.
\end{IEEEproof}

Proposition~\ref{prop:sfc-extension} translates the mechanism into SFC terms. The harmful option is a service chain that couples security resources whose congestion externalities were previously separated across distributed chains. The topology experiments in Section~\ref{sec:evaluation} use the same pattern: the zero-trust gateway path traverses shared ingress and egress security planes and changes both load concentration and attacker target value.

\section{Paradox-Aware Orchestration}
\label{sec:algorithm}

The mitigation principle is to evaluate defensive options as equilibrium-shaping resources before exposing them to normal service-chain selection. A new option may still be useful for failover, incident response, or selected tenants, but its use should be constrained when the unrestricted equilibrium is Braessian.

\begin{algorithm}[t]
\caption{Paradox-Aware Defensive Option Screening}
\label{alg:screening}
\begin{algorithmic}[1]
\REQUIRE Baseline system $\mathcal{M}$, candidate action $\Delta$, penalty threshold $\tau$, cap grid $\mathcal{K}$
\STATE Compute baseline equilibrium $f^\star$ and cost $J(\mathcal{M})$
\STATE Construct candidate system $\mathcal{M}'=\mathcal{M}+\Delta$
\STATE Compute unrestricted equilibrium $f'^\star$ and penalty $\Pi(\Delta)$
\IF{$\Pi(\Delta)\le \tau$}
    \STATE expose $\Delta$ without restriction
\ELSE
    \FOR{cap fraction $\kappa\in\mathcal{K}$ from large to small}
        \STATE constrain each new-option path by $f_p\le \kappa d_r$
        \STATE recompute equilibrium and penalty $\Pi_\kappa$
        \IF{$\Pi_\kappa\le \tau$}
            \STATE expose $\Delta$ under cap $\kappa$; \textbf{return}
        \ENDIF
    \ENDFOR
    \STATE apply operator fallback: reserve $\Delta$ for failover or hide it from normal routing
\ENDIF
\end{algorithmic}
\end{algorithm}

The algorithm leaves the security function deployed and controls its visibility to ordinary chain selection. A fallback has two common forms: reserve the option for failover, in which case it remains available during incidents but is absent from ordinary path selection; or hide it from normal placement until the traffic matrix or policy changes. In the worst case, the screen solves $|\mathcal{K}|$ convex equilibrium problems, one per candidate cap level. The screen runs offline before deployment, and the implementation uses Frank--Wolfe all-or-nothing assignments over each tenant request with exact line search for affine resource delays.

In an SDN/NFV deployment, the actions in Algorithm~\ref{alg:screening} map to standard control-plane mechanisms. Unrestricted exposure corresponds to admitting the new chain in the service-function classifier. A cap can be implemented by weighted next-hop groups, rate limits, tenant-specific policy rules, or admission-control quotas on the gateway chain. Reserving the option for failover means the gateway remains deployed but absent from ordinary path computation until a failure or incident policy activates it. Thus the algorithm is not an online packet scheduler; it is a pre-deployment policy screen that decides how much of a defensive option should be made visible to normal orchestration.

\section{Experimental Evaluation}
\label{sec:evaluation}

\subsection{Topology-Derived SFC Suite}

The experiments use four deterministic topology-derived scenarios:
\begin{itemize}[leftmargin=*,topsep=2pt,itemsep=2pt]
    \item \textbf{Fat-tree datacenter:} a $k=4$ fat-tree with 36 nodes and 48 links.
    \item \textbf{NSFNET-style WAN:} a 14-node, 19-link backbone.
    \item \textbf{GEANT-style WAN:} a 22-node, 35-link European backbone abstraction.
    \item \textbf{Edge/fog:} regional edge clusters connected through fog and cloud-core nodes.
\end{itemize}

For each topology, peripheral nodes are selected as tenant sources and destinations and centrality-ranked nodes are used as candidate VNF sites. Before intervention, each request can use two distributed FW--IDS--WAF chains. The defensive intervention adds a shared zero-trust gateway path with lower free-flow delay and higher exposure. Link resources inherit latency and capacity from the topology generator; VNF resources include fixed processing delay, load sensitivity, capacity, risk, and exposure.

\input{tables/experiment_parameters}

\subsection{Implementation Details}

The experiment driver constructs a multi-commodity SFC instance for each topology and solves the post-policy equilibrium using Frank--Wolfe iterations. At each iteration, every tenant request performs an all-or-nothing assignment to the least-cost feasible chain under the current resource loads; exact line search is available because all resource delays are affine. This traffic-engineering abstraction lets us evaluate the equilibrium induced by a management action before a specific firewall, IDS, or gateway implementation is selected.

The topology generator assigns link latency, capacity, risk, and exposure from structural roles in the graph. Centrality-ranked nodes host candidate security functions, while peripheral nodes generate tenant demands. The experiments do not measure a specific deployed gateway; they embed the same SFC construction into datacenter, WAN, and edge/fog substrates with different path diversity and centrality patterns. The default parameters instantiate the Braessian regime from Section~\ref{sec:theory}. The sensitivity analysis shows where the gateway becomes beneficial or non-Braessian.

\subsection{Baselines and Policies}

We compare seven policies under the same equilibrium and attack-loss evaluation:
\begin{enumerate}[leftmargin=*,topsep=2pt,itemsep=2pt]
    \item \textbf{No expansion:} original distributed chains only.
    \item \textbf{Naive expansion:} expose the new gateway to every tenant without restriction.
    \item \textbf{Load-aware cap:} cap each tenant's gateway-path flow at 50\% of demand.
    \item \textbf{Risk-aware surcharge:} add a static $\rho_e\chi_e$ surcharge to path scoring, then evaluate the selected flow under the original delay and loss model.
    \item \textbf{Min--max utilization cap:} choose the gateway cap that minimizes the maximum resource utilization over the cap grid.
    \item \textbf{System optimum:} centrally minimize total delay after the gateway is added, using marginal-cost routing.
    \item \textbf{Paradox-aware:} run Algorithm~\ref{alg:screening} with $\tau=0.02$ and choose the loosest cap satisfying the threshold.
\end{enumerate}

The risk-aware baseline optimizes an auxiliary local score but is evaluated under the same service and attack-loss model as all other policies. Risk-aware scoring can avoid exposed resources even when they have low latency; min--max utilization can hide the gateway if any use raises the peak load. These conservative baselines test whether paradox-aware orchestration can preserve some defensive availability while keeping the equilibrium penalty below the operator's threshold.

\input{tables/topology_summary}

\begin{figure*}[t]
\centering
\includegraphics[width=0.94\textwidth]{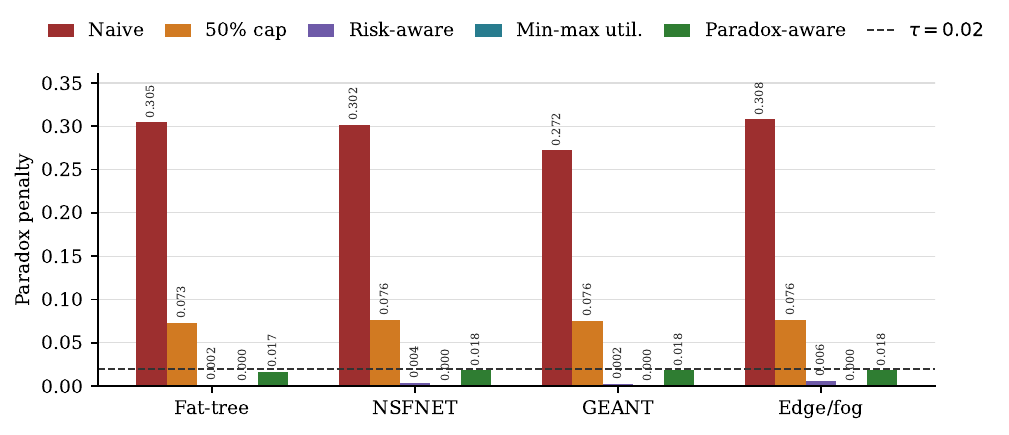}
\caption{Paradox penalty across policies and topologies. Naive expansion is consistently Braessian, while risk-aware and min--max-utilization baselines often avoid the gateway entirely. Paradox-aware orchestration keeps the penalty below $\tau=0.02$ while preserving capped gateway access.}
\label{fig:policy-penalties}
\end{figure*}

Table~\ref{tab:topology-summary} and Fig.~\ref{fig:policy-penalties} report the main cross-topology result. Naive expansion is Braessian in every topology, with penalties between 27.2\% and 30.8\%. The Security Braess Ratio ranges from 1.272 to 1.308. Paradox-aware constrained use keeps the residual penalty below 1.9\% while preserving limited access to the added security gateway.

\input{tables/policy_comparison}

Table~\ref{tab:policy-comparison} shows the policy-level tradeoff. The 50\% load-aware cap reduces the naive penalty but remains 7.3--7.6\% above the baseline. Risk-aware scoring and min--max utilization are safer on service cost because they assign little or no ordinary traffic to the gateway; they act like conservative suppression policies. The system-optimum baseline eliminates the penalty, confirming that the added gateway is not physically harmful when a centralized controller internalizes congestion externalities. Paradox-aware screening permits a 25\% gateway share in the default setting, keeps the Wardrop penalty around 1.7--1.8\%, and reduces expected attack loss by preventing full concentration on the new gateway. The similar aware penalties across topologies reflect the binding threshold $\tau=0.02$ and the discrete cap grid, not a claim of topology-invariant residual harm. We use gateway share as a defensive-availability metric: unlike suppression baselines, paradox-aware orchestration keeps the new security function available to ordinary traffic under a bounded equilibrium penalty.

\subsection{Attack Stress}

The attack stress test targets resources after equilibrium and reports service-loss proxies. The deterministic DDoS column targets the most loaded resource, asking whether orchestration creates a high-flow chokepoint that an adaptive attacker would prefer. To reduce coupling between the metric and the mitigation, Table~\ref{tab:attack-stress} also reports a weighted attacker that samples targets with probability proportional to utilization and exposure. Naive expansion raises both stress-loss measures in every topology. Paradox-aware capping reduces the losses while keeping the security option available for limited use.

\input{tables/attack_stress_summary}

\begin{figure*}[t]
\centering
\includegraphics[width=0.94\textwidth]{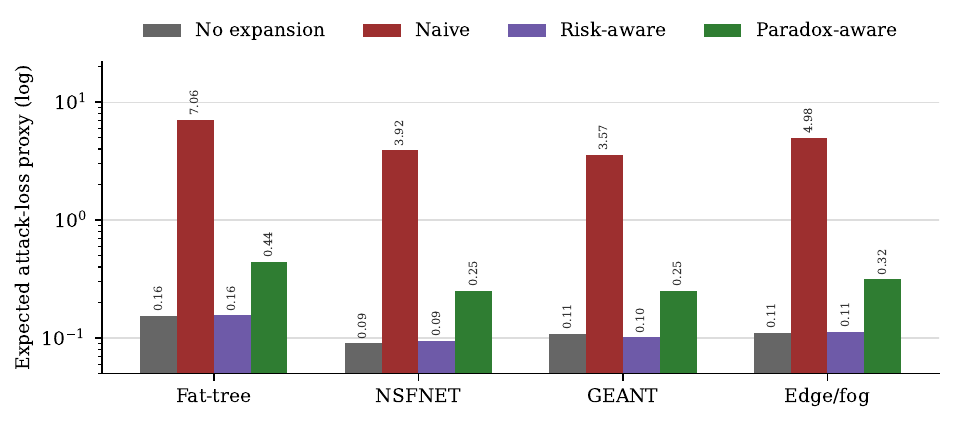}
\caption{Expected attack-loss proxy after service equilibrium. Naive exposure creates a high-value target in every topology. Paradox-aware capping reduces the target value while still allowing controlled gateway use; risk-aware scoring is conservative and often avoids the gateway entirely.}
\label{fig:attack-loss}
\end{figure*}

Fig.~\ref{fig:attack-loss} shows the expected targeted loss proxy. The largest reductions occur because the gateway no longer receives nearly all tenant demand. The metric is a concentration-sensitive stress measure, not a calibrated prediction of dollars lost or packets dropped.

\subsection{Nonlinear-Curve Robustness}

The main model uses affine delay to retain a convex equilibrium program and closed-form line search. Table~\ref{tab:queueing-robustness} re-evaluates the baseline, naive, and paradox-aware flows under a steeper BPR-style nonlinear delay curve,
\[
\hat{\ell}_e(y)=b_e+a_e\frac{y}{u_e}\left(1+0.15\left(\frac{y}{u_e}\right)^4\right),
\]
with utilization clipped at $1.5$ for finite stress evaluation. The calculation does not replace a measured VNF service curve; it tests whether the same flow assignments remain harmful under a nonlinear congestion penalty. The ordering is unchanged: naive exposure has the largest nonlinear-curve cost, while paradox-aware capping remains close to the baseline.

\input{tables/queueing_robustness}

\subsection{Sensitivity Analysis}

Table~\ref{tab:sensitivity} reports one-factor sweeps on the NSFNET-style topology. The paradox is strongest when the gateway is highly attractive in free-flow delay and when shared security planes are sufficiently load-sensitive. When load sensitivity is too low, the new gateway can be beneficial rather than Braessian; when the threshold $\tau$ is relaxed, the algorithm exposes a larger gateway share and the residual penalty rises by design. Gateway capacity and risk/exposure scaling affect only the post-equilibrium attack-loss proxy, not the Wardrop service cost, because the default gateway has zero service-delay slope ($a_g=0$): its capacity enters the evaluation only through the overload term of the loss proxies.

\input{tables/sensitivity_summary}

\begin{figure*}[t]
\centering
\includegraphics[width=0.94\textwidth]{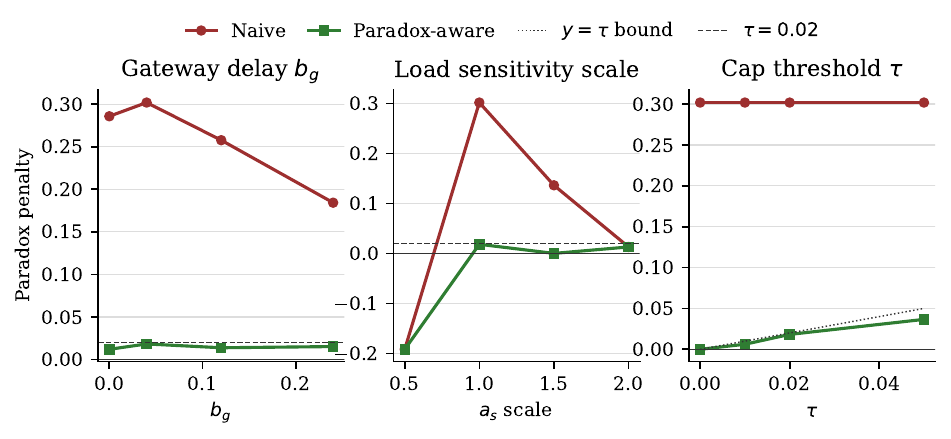}
\caption{Sensitivity of the paradox penalty on the NSFNET-style topology. The dashed line is the default screening threshold $\tau=0.02$. The paradox is a regime property: it appears when the gateway is sufficiently attractive and when shared security resources are sufficiently load-sensitive.}
\label{fig:sensitivity-phase}
\end{figure*}

Fig.~\ref{fig:sensitivity-phase} shows the phase behavior. Increasing gateway delay eventually weakens the paradox because fewer flows find the gateway attractive. Reducing shared-resource load sensitivity can make the gateway beneficial, producing negative penalty; very high sensitivity can also make the shortcut unattractive enough that unrestricted use weakens. Raising $\tau$ does not change the naive equilibrium; it changes how much residual penalty the paradox-aware controller tolerates in exchange for gateway availability. This threshold is the primary tuning parameter in the screening method.

\subsection{Reproducibility}

The experiment code is available at \url{https://github.com/dcommey/braess-paradox-in-networks}. The driver regenerates the result tables, figures, and regression tests reported here.

\section{Discussion}

Three management implications follow from the experiments.

First, defensive option expansion should not be evaluated only by free-flow latency or a local inspection-benefit score. In all four topology-derived scenarios, the new zero-trust gateway is locally attractive, but unrestricted use increases equilibrium cost by roughly 30\%.

Second, load awareness alone is insufficient. A generic 50\% cap reduces congestion but still leaves nontrivial residual penalty and high risk concentration. The appropriate cap depends on the post-adaptation equilibrium.

Third, paradox-aware control does not conflict with defense-in-depth. The algorithm does not delete a gateway or disable a security VNF; it exposes the option under rate limits, reserves it for failover, or restricts it to selected chains.

\subsection{Design Guidance}

In deployment, candidate security actions should first be evaluated under their unrestricted post-adaptation equilibrium, along with local latency, capacity, and inspection benefit. If the unrestricted equilibrium is safe, the action can be exposed normally. If it is Braessian, the operator can test constrained exposure: per-tenant caps, class-specific eligibility, failover-only activation, or randomized access. If all constrained policies remain unsafe, the option should be hidden from ordinary orchestration.

The risk-aware and min--max-utilization baselines in Table~\ref{tab:policy-comparison} often avoid the gateway almost entirely, which is safe but overly cautious. The system optimum shows that the gateway can be harmless under a controller that prices the congestion externality. Paradox-aware screening occupies the middle ground: the operator may lack or be unable to enforce a perfect system optimum, but still wants to expose the defensive option to some traffic rather than suppress it.

The cap threshold $\tau$ is an operator policy parameter rather than a universal constant. Lower $\tau$ values prioritize strict service resilience and may reserve the gateway for incidents only. Higher $\tau$ values allow more ordinary gateway use at the cost of residual equilibrium penalty. Gateway free-flow delay and shared-resource load sensitivity are design-time signals: attractive shortcuts through highly load-sensitive planes should be screened before they are advertised to tenants or automated chain-placement policies.

\subsection{Path to Deployment}

A production orchestrator would not need to solve the screening problem for every packet or flowlet. It would run the screen when a new VNF instance, gateway, backup route, or migration target is introduced; when a traffic matrix changes materially; or when a policy update changes which tenants can use a chain. The output can be encoded as ordinary control-plane state: service-function classifier rules, path weights, admission quotas, failover flags, or tenant-specific eligibility constraints. Existing SFC/NFV control loops can therefore use the screen as an offline policy step.

\subsection{Limitations}

The evaluation uses affine delay functions to keep the equilibrium program convex and the sufficient condition interpretable; production deployments may require BPR-style latency curves, queueing models, or measured VNF service curves. The adversary is evaluated sequentially after the service equilibrium, so the model does not solve a joint defender-attacker game. Demand is deterministic within each experiment, and the topology parameters are synthetic rather than measured from deployed firewall, IDS, WAF, or gateway appliances. These choices isolate the management-level non-monotonicity from implementation-specific effects. The screening framework requires only a callable equilibrium solver, so measured service curves, multi-domain policy constraints, stochastic traffic matrices, or calibrated attacker models can replace the current inputs when available.

\section{Conclusion}

Adding a locally attractive defensive option to an SFC system can worsen the equilibrium by concentrating traffic and adversarial value on shared security resources. The sufficient condition in this paper identifies one such security-induced Braess paradox, and the screening algorithm caps or reserves harmful options before deployment. In topology-derived datacenter, WAN, and edge/fog experiments, naive expansion raised equilibrium cost by up to 30.8\%, while paradox-aware orchestration kept the residual penalty below 1.9\%. Defensive options should be managed as security controls and as equilibrium-shaping resources.

\bibliographystyle{IEEEtran}
\bibliography{references}

\end{document}

%% file: tables/experiment_parameters.tex
% Auto-generated by scripts/run_experiments.py
\begin{table}[t]
\centering
\caption{Default experiment parameters.}
\label{tab:experiment-parameters}
\begin{tabular}{ll}
\toprule
Parameter & Value \\
\midrule
Tenant demand $d_r$ & 1.0 for every request \\
Gateway free-flow delay $b_g$ & 0.04 \\
Gateway load slope $a_g$ & 0.0 \\
Gateway capacity factor & 1.0 times total demand \\
Shared security-plane slope & 1.0 \\
Distributed IDS delay & 1.04--1.076 \\
Link base delay & $0.025\times$ topology latency \\
Link load slope & $0.02\times$ topology latency \\
Gateway risk/exposure & $\rho_g=0.42$, $\chi_g=2.4$ \\
Shared-plane risk/exposure & $\rho=0.08$, $\chi=1.1$ \\
Penalty threshold $\tau$ & 0.02 \\
Risk-aware surcharge weight & 1.0 \\
Cap grid $\mathcal{K}$ & 1.0, 0.75, 0.5, 0.35, 0.25, 0.15, 0.1, 0.0 \\
Frank--Wolfe tolerance & $10^{-5}$ relative gap \\
\bottomrule
\end{tabular}
\end{table}

%% file: tables/topology_summary.tex
% Auto-generated by scripts/run_experiments.py
\begin{table*}[t]
\centering
\caption{Topology-derived SFC/NFV results. Paradox-aware values use pre-deployment screening and constrained use of the added gateway.}
\label{tab:topology-summary}
\begin{tabular}{lrrrrrrr}
\toprule
Topology & Nodes & Edges & Requests & Naive penalty & Naive SBR & Aware penalty & Aware gain \\
\midrule
Fat-tree & 36 & 48 & 7 & 0.305 & 1.305 & 0.017 & 22.1\% \\
NSFNET & 14 & 19 & 4 & 0.302 & 1.302 & 0.018 & 21.8\% \\
GEANT & 22 & 35 & 4 & 0.272 & 1.272 & 0.018 & 20.0\% \\
Edge/fog & 25 & 30 & 5 & 0.308 & 1.308 & 0.018 & 22.1\% \\
\bottomrule
\end{tabular}
\end{table*}

%% file: tables/policy_comparison.tex
% Auto-generated by scripts/run_experiments.py
\begin{table*}[t]
\centering
\scriptsize
\caption{Policy comparison under the same adaptive equilibrium model.}
\label{tab:policy-comparison}
\begin{tabular}{llrrrrrr}
\toprule
Topology & Policy & Cost & Penalty & Risk conc. & Attack loss & Max util. & Gateway share \\
\midrule
Fat-tree & No expansion & 1.581 & 0.000 & 0.083 & 0.155 & 0.502 & 0.00 \\
Fat-tree & Naive expansion & 2.062 & 0.305 & 0.633 & 7.056 & 1.000 & 1.00 \\
Fat-tree & Load-aware cap & 1.697 & 0.073 & 0.451 & 1.764 & 0.752 & 0.50 \\
Fat-tree & Risk-aware & 1.584 & 0.002 & 0.088 & 0.155 & 0.502 & 0.00 \\
Fat-tree & Min-max util. & 1.581 & 0.000 & 0.083 & 0.155 & 0.502 & 0.00 \\
Fat-tree & System optimum & 1.580 & -0.000 & 0.080 & 0.164 & 0.517 & 0.03 \\
Fat-tree & Paradox-aware & 1.607 & 0.017 & 0.273 & 0.441 & 0.627 & 0.25 \\
NSFNET & No expansion & 1.587 & 0.000 & 0.101 & 0.090 & 0.506 & 0.00 \\
NSFNET & Naive expansion & 2.066 & 0.302 & 0.619 & 3.917 & 1.000 & 0.99 \\
NSFNET & Load-aware cap & 1.708 & 0.076 & 0.458 & 1.008 & 0.755 & 0.50 \\
NSFNET & Risk-aware & 1.593 & 0.004 & 0.105 & 0.094 & 0.518 & 0.00 \\
NSFNET & Min-max util. & 1.587 & 0.000 & 0.101 & 0.090 & 0.506 & 0.00 \\
NSFNET & System optimum & 1.586 & -0.000 & 0.091 & 0.094 & 0.518 & 0.03 \\
NSFNET & Paradox-aware & 1.616 & 0.018 & 0.290 & 0.252 & 0.631 & 0.25 \\
GEANT & No expansion & 1.595 & 0.000 & 0.091 & 0.109 & 0.514 & 0.00 \\
GEANT & Naive expansion & 2.029 & 0.272 & 0.581 & 3.573 & 0.979 & 0.94 \\
GEANT & Load-aware cap & 1.716 & 0.076 & 0.421 & 1.008 & 0.764 & 0.50 \\
GEANT & Risk-aware & 1.599 & 0.002 & 0.097 & 0.101 & 0.536 & 0.00 \\
GEANT & Min-max util. & 1.595 & 0.000 & 0.091 & 0.109 & 0.514 & 0.00 \\
GEANT & System optimum & 1.592 & -0.002 & 0.122 & 0.106 & 0.550 & 0.08 \\
GEANT & Paradox-aware & 1.624 & 0.018 & 0.257 & 0.252 & 0.639 & 0.25 \\
Edge/fog & No expansion & 1.577 & 0.000 & 0.050 & 0.110 & 0.501 & 0.00 \\
Edge/fog & Naive expansion & 2.063 & 0.308 & 0.487 & 4.979 & 1.000 & 0.99 \\
Edge/fog & Load-aware cap & 1.698 & 0.076 & 0.298 & 1.260 & 0.751 & 0.50 \\
Edge/fog & Risk-aware & 1.587 & 0.006 & 0.053 & 0.113 & 0.506 & 0.00 \\
Edge/fog & Min-max util. & 1.577 & 0.000 & 0.050 & 0.110 & 0.501 & 0.00 \\
Edge/fog & System optimum & 1.577 & -0.000 & 0.048 & 0.115 & 0.511 & 0.02 \\
Edge/fog & Paradox-aware & 1.606 & 0.018 & 0.155 & 0.315 & 0.626 & 0.25 \\
\bottomrule
\end{tabular}
\end{table*}

%% file: tables/attack_stress_summary.tex
% Auto-generated by scripts/run_experiments.py
\begin{table}[t]
\centering
\caption{Post-equilibrium attack stress tests. DDoS targets the most loaded resource; weighted loss randomizes targets proportional to utilization and exposure.}
\label{tab:attack-stress}
\begin{tabular}{lrrrrr}
\toprule
Topology & DDoS-B & DDoS-N & DDoS-A & W-N & W-A \\
\midrule
Fat-tree & 0.682 & 2.510 & 0.852 & 3.114 & 0.114 \\
NSFNET & 0.688 & 1.360 & 0.858 & 1.754 & 0.083 \\
GEANT & 0.699 & 1.332 & 0.869 & 1.512 & 0.087 \\
Edge/fog & 0.681 & 1.360 & 0.851 & 1.616 & 0.073 \\
\bottomrule
\end{tabular}
\end{table}

%% file: tables/queueing_robustness.tex
% Auto-generated by scripts/run_experiments.py
\begin{table}[t]
\centering
\caption{Robustness under nonlinear delay evaluation.}
\label{tab:queueing-robustness}
\begin{tabular}{lrrr}
\toprule
Topology & Baseline & Naive & Aware \\
\midrule
Fat-tree & 1.586 & 2.362 & 1.625 \\
NSFNET & 1.591 & 2.353 & 1.634 \\
GEANT & 1.600 & 2.281 & 1.642 \\
Edge/fog & 1.582 & 2.358 & 1.624 \\
\bottomrule
\end{tabular}
\end{table}

%% file: tables/sensitivity_summary.tex
% Auto-generated by scripts/run_experiments.py
\begin{table*}[t]
\centering
\caption{One-factor sensitivity analysis on the NSFNET-style topology.}
\label{tab:sensitivity}
\begin{tabular}{llrrrrr}
\toprule
Parameter & Value & Naive penalty & Aware penalty & Aware gain & Naive attack & Cap \\
\midrule
$b_g$ & 0.00 & 0.286 & 0.012 & 21.3\% & 4.032 & 0.25 \\
$b_g$ & 0.04 & 0.302 & 0.018 & 21.8\% & 3.917 & 0.25 \\
$b_g$ & 0.12 & 0.258 & 0.014 & 19.4\% & 2.818 & 0.15 \\
$b_g$ & 0.24 & 0.184 & 0.015 & 14.3\% & 1.454 & 0.10 \\
$u_g$ factor & 0.50 & 0.302 & 0.018 & 21.8\% & 7.721 & 0.25 \\
$u_g$ factor & 0.75 & 0.302 & 0.018 & 21.8\% & 5.148 & 0.25 \\
$u_g$ factor & 1.00 & 0.302 & 0.018 & 21.8\% & 3.917 & 0.25 \\
$u_g$ factor & 1.50 & 0.302 & 0.018 & 21.8\% & 3.917 & 0.25 \\
$a_s$ factor & 0.50 & -0.192 & -0.192 & 0.0\% & 4.032 & 1.00 \\
$a_s$ factor & 1.00 & 0.302 & 0.018 & 21.8\% & 3.917 & 0.25 \\
$a_s$ factor & 1.50 & 0.136 & 0.000 & 12.0\% & 0.469 & 0.00 \\
$a_s$ factor & 2.00 & 0.013 & 0.013 & 0.0\% & 0.094 & 1.00 \\
$\tau$ & 0.00 & 0.302 & 0.000 & 23.2\% & 3.917 & 0.00 \\
$\tau$ & 0.01 & 0.302 & 0.006 & 22.7\% & 3.917 & 0.15 \\
$\tau$ & 0.02 & 0.302 & 0.018 & 21.8\% & 3.917 & 0.25 \\
$\tau$ & 0.05 & 0.302 & 0.037 & 20.4\% & 3.917 & 0.35 \\
$\rho_g,\chi_g$ scale & 0.50 & 0.302 & 0.018 & 21.8\% & 0.979 & 0.25 \\
$\rho_g,\chi_g$ scale & 1.00 & 0.302 & 0.018 & 21.8\% & 3.917 & 0.25 \\
$\rho_g,\chi_g$ scale & 1.50 & 0.302 & 0.018 & 21.8\% & 8.813 & 0.25 \\
$\rho_g,\chi_g$ scale & 2.00 & 0.302 & 0.018 & 21.8\% & 15.668 & 0.25 \\
\bottomrule
\end{tabular}
\end{table*}